\documentclass[%
 aip,
 amsmath,amssymb,
reprint, onecolumn 
]{revtex4-1}

\usepackage{graphicx}
\usepackage{dcolumn}
\usepackage{bm}

\usepackage[utf8]{inputenc}
\usepackage[T1]{fontenc}
\usepackage{mathptmx}
\usepackage{etoolbox}
\usepackage{xcolor}
\usepackage{amsthm}

\newtheorem{proposition}{Proposition}
\newtheorem{lemma}{Lemma}

\makeatletter
\def\@email#1#2{%
 \endgroup
 \patchcmd{\titleblock@produce}
  {\frontmatter@RRAPformat}
  {\frontmatter@RRAPformat{\produce@RRAP{*#1\href{mailto:#2}{#2}}}\frontmatter@RRAPformat}
  {}{}
}%
\makeatother
\begin{document}

\title{Metric Congruence in Finite-Dimensional Non-Hermitian Quantum Mechanics}
\author{Romina Ram\'\i rez}
 \email{romina@mate.unlp.edu.ar}
 \affiliation{IAM, CONICET--CeMaLP, University of La Plata, Argentina.}
\author{Marta Reboiro}%
 \email{reboiro@fisica.unlp.edu.ar}
\affiliation{IFLP, CONICET--Department of Physics, University of La Plata, Argentina}%

\date{\today}

\begin{abstract}
We study metric representations in finite-dimensional non-Hermitian quantum mechanics. 
The main purpose of this work is to emphasize that the physical description of a 
non-Hermitian system may be formulated in different, but isomorphic, Hilbert spaces. 
In particular, within the Krein-space formalism, we show that the vector space endowed with an indefinite Krein metric can be explicitly related to the standard Hilbert space through a suitable isomorphism.

This observation is essential for a consistent description of non-Hermitian Hamiltonians. 
Physical states, metrics, and operators must be transported through the corresponding Hilbert-space isomorphism. In this way, equivalent representations of the same system can be used without changing the physical content of the theory.

We illustrate these theoretical aspects by studying a two-level non-Hermitian spin model. We use the Robertson uncertainty relation as a consistency test. Apparent violations can arise when operators and states are kept fixed while the metric is changed, and therefore reflect a mismatch of representations rather than a failure of the uncertainty principle.

\end{abstract}

\maketitle
\section{Introduction}

Non-Hermitian Hamiltonians occur in several quantum settings  \cite{Ashida02072020,xue1,epjd2020,Elganainy}. In open dynamics, they are often used as effective generators  \cite{nori,chen2021quantum,joglekar}.

When these Hamiltonians are pseudo-Hermitian, exceptional points \cite{kato1966perturbation,Heiss_2012} and dynamical phase transitions take place \cite{pastawski,Bergholtz,Rahmani,Fu}.

The isomorphism between two Hilbert spaces \(\mathcal{H}_{1}\) and
\(\mathcal{H}_{2}\), endowed with the metric operators \(\eta_1\) and
\(\eta_2\), respectively, means the existence of an invertible operator
\(T:\mathcal{H}_{1}\to \mathcal{H}_{2}\) such that
\[
\langle Tx,\eta_2 T y\rangle
=
\langle x,\eta_1 y\rangle,
\qquad
\forall x,y\in\mathcal{H}_{1}.
\]
Equivalently, $T^\dagger \eta_2 T=\eta_1$. Changing the metric without changing the representatives of
states and operators mixes different Hilbert-space structures. This is
the source of many apparent inconsistencies in the treatment of non-Hermitian models \cite{mandel}. The simultaneous transport ensures the equivalence between Hilbert spaces defined by inner products induced by $\eta_{1}$ and $\eta_{2}$.

Several prescriptions are used in practice. Some authors use the original
inner product and treat the non-Hermitian Hamiltonian as an effective
generator of nonunitary evolution, as in open-system or postselected
dynamics \cite{joglekar}. Other authors use the biorthogonal
description. If the Hamiltonian is diagonalizable, right and left
eigenvectors provide a biorthogonal resolution of the identity. When the
spectrum is real, this structure can be used to construct a positive
metric with respect to which the Hamiltonian is self-adjoint
\cite{brody2013biorthogonal,BenderBrodyJones2002}. In that case, the biorthogonal and quasi-Hermitian formulations are closely related.

Other prescriptions are based on the explicit construction of a physical
metric operator. This includes formulations in which an additional operator,
usually denoted by \(\mathcal C\), is used to define a positive inner product
\cite{BenderBrodyJones2002}. The situation becomes more
subtle when the metric is time dependent, since in that case the metric,
states, observables, and dynamical generator must be treated consistently
within the same time-dependent Hilbert-space representation \cite{Mostafazadeh2007TimeDep,fringtm1}.

For a pseudo-Hermitian Hamiltonian, there exists a similarity operator $S=S^{\dagger}$ such that $SH= H^{\dagger}S$. Depending on the model parameter space associated with a given Hamiltonian, the spectrum consists of real eigenvalues or the appearance of complex conjugate pairs. They correspond to the so-called exact symmetry phase and spontaneously broken symmetry phase. The boundary between the two phases consists of exceptional points.

In the exact-symmetry phase, one can construct a positive-definite metric operator $S$ obeying $SH=H^{\dagger} S$, which can be taken as a metric inducing a new Hilbert space. This is not the case for the spontaneously broken-symmetry regime and exceptional points. Starting from the similarity operator $S$, the metric can be defined by using the Krein formalism.

If
\begin{equation}
    S=VDV^\dagger ,
\end{equation}
with \(D\) real and nonsingular, we define
\begin{equation}
    K=V \tilde{D}V^\dagger .
\end{equation}
where $\tilde{D}$ is the absolute value of the diagonal matrix $D$.
This Krein-positivized operator \(K\) defines a Hilbert-space
metric. However, 
\begin{equation}
    KH \neq H^\dagger K .
\end{equation}

In finite dimension, any two metrics are related by a
congruence transformation:
\begin{equation}
    \eta_1=T^\dagger \eta_2 T.
\end{equation}
To be consistent, under this change states and operators must transform as
\begin{equation}
    x\mapsto Tx,
    \qquad
    A\mapsto TAT^{-1}.
\end{equation}

We shall use uncertainty relations to test this consistency requirement.
The Heisenberg--Robertson relation bounds from below the product of variances
by the modulus of the expectation value of the commutator
\cite{Heisenberg1927,Robertson}. Schrödinger's refinement strengthens this
bound by incorporating the covariance contribution \cite{Schrodinger1930}. Product-form relations can become trivial when one
variance vanishes, even if the observables are not jointly sharp on the
state. This limitation motivates sum-variance relations \cite{MacconePati2014}. These stronger relations have also been investigated in experimental and operational
contexts, showing their relevance beyond purely formal considerations
\cite{uncert,exp2URNH,URNH}.

In this direction, the authors of Ref.\cite{mandel} introduced the notion of ``good observables'' in finite-dimensional non-Hermitian systems endowed with a $G$-metric inner product \cite{mandel}. In their formulation, an operator $O$ is physically admissible when it satisfies $O^\dagger G=GO$, and the corresponding uncertainty relation is valid for pairs of such metric-compatible observables. From the point of view developed here, this condition can be understood geometrically as the self-adjointness condition in the Hilbert space defined by the metric $G$.

Although the isomorphism  between finite-dimensional Hilbert spaces is a well-known concept in linear algebra, the scope of this work is to provide a way to compare the abundant number of metrics adopted in the literature around the treatment of non-Hermitian Hamiltonians and the consistency of physical results.

The paper is organized as follows. Section~\ref{formalismo} fixes the
notation for metric representations and recalls the construction of
intertwiners in different spectral regimes. Section~\ref{theorem}
proves the transport rule for metrics related by congruence. In particular, this section studies the relation between an indefinite intertwiner \(S\), its positive transformed operator \(K\), and the condition for \(K\) to be a quasi-Hermitian metric for the Hamiltonian. Operator $K$ allows the construction of a new Hilbert space $\mathcal{H}_{K}$. We show that physical results computed in this new Hilbert space are equivalent to the ones computed in the original Hilbert space $\mathcal{H}$.

Section~\ref{two-level-model} applies the formalism to a two-level
non-Hermitian spin model and distinguishes the unbroken, broken, and
exceptional-point regimes. Section~\ref{uncertainty} analyzes
Robertson and sum-variance uncertainty relations under metric transport. Conclusions are given in Section~\ref{conclusiones}.

\section{Formalism}\label{formalismo}

Consider a finite dimensional Hilbert space  \(\mathcal H\simeq\mathbb C^n\)  equipped with the canonical inner product $\langle .,.\rangle$.
Let $A$ be an operator for which there exists a self-adjoint operator $\eta$ such that 
\begin{equation}
    \eta A = A^{\dagger}\eta 
\end{equation}
This condition implies that $A$ is a pseudo-Hermitian operator.
If $\eta$ is  positive definite, i.e $\eta= \rho^{\dagger}\rho$, it can be used as a metric to define a new inner product
\begin{equation}\label{eta}
    \langle .,. \rangle_{\eta}=\langle .,\eta . \rangle
\end{equation}
In this case, operator $A$ is self-adjoint in $\mathcal{H}_{\eta}= (\mathcal{H}, \langle .,. \rangle_{\eta})$ where

\begin{equation}
    \langle ., A. \rangle_{\eta} =\langle ., \eta A. \rangle = \langle ., A^{\dagger } \eta. \rangle =\langle A ., . \rangle_{\eta}.
\end{equation}
It is straightforward to show that there exists an isomorphism between this Hilbert space and the original one, provided that the operator $A$ is related to a Hermitian operator  
\begin{equation}
    a=\rho A \rho^{-1}
\end{equation}
in the original Hilbert space \cite{Mostafazadeh1,Mostafazadeh2,Mostafazadeh3}. 

This observation also clarifies the meaning of the good-observable condition $O^\dagger G=GO$ used in Ref.~\cite{mandel}. If $G=\rho^\dagger \rho$, then an operator $O$ satisfying this condition is self-adjoint in the metric Hilbert space $\mathcal H_G$ and is mapped by the isomorphism to the Hermitian operator $\rho O\rho^{-1}$ in the reference Hilbert space. Conversely, every Hermitian observable in the reference representation defines a good observable after the inverse transport.

This is consistent with the standard transformation rule under a change
of basis. When passing from one basis to another, states and operators
must be represented in the same basis. Hence, any operator defined in the
initial representation has to be transformed to the new representation
before its adjointness properties, expectation values or observable
character can be evaluated.

For a pseudo-Hermitian Hamiltonian in the parameter region corresponding to the exact-symmetry phase, it is possible to construct a positive-definite metric operator.

In the spontaneously broken-symmetry regime and at exceptional points it is not possible to define a positive $\eta$ that serves as a similarity operator between $A$ and $A^{\dagger}$.

In such cases, there is no unique canonical prescription in the literature. Instead, different prescriptions have been adopted.

One possible prescription is to keep the initial Hilbert-space structure
and study the no-jump effective non-Hermitian Hamiltonian associated with a Lindblad dynamics \cite{nori,joglekar}. In this scheme,  the trace of the density matrix is no longer preserved and a renormalization is needed at each instant of time. 

Another possibility is to consider the so-called biorthonormal framework \cite{brody2013biorthogonal,paolo}. In this approach, a diagonalizable non-Hermitian operator $A$ is described by right and left eigenvectors,
\begin{equation}
            A |R_n\rangle = E_n |R_n\rangle,
        \qquad
        A^\dagger |L_n\rangle = E_n^* |L_n\rangle ,
\end{equation}

normalized according to
\begin{equation}\label{ort}
        \langle L_m|R_n\rangle=\delta_{mn}.
\end{equation}
with the resolution of the identity
\begin{equation}
        \sum_n |R_n\rangle\langle L_n|=I,
\end{equation}

Thus, any state in $\mathcal{H}$ can be expanded as
\begin{equation}
        |\psi\rangle=\sum_n \langle L_n|\psi\rangle |R_n\rangle .
\end{equation}
The inner product calculated in Eq.(\ref{ort}), can be obtained by introducing the operator 
\begin{equation}\label{S}
S = \sum |L_n\rangle \langle L_n |
\end{equation}
such that $\langle L_m|R_n\rangle=\langle R_m|R_n\rangle_{S}=\delta_{mn}$. 

The operator \(S\) defines a suitable metric for the treatment of
pseudo-Hermitian Hamiltonians with real spectra, and can also be considered
for operators that do not belong to the pseudo-Hermitian class.
It therefore defines the metric Hilbert space
\[
        \mathcal H_{S}
        =
        \bigl(\mathcal H,\langle\cdot,\cdot\rangle_{S}\bigr),
        \qquad
        \langle f,g\rangle_{S}
        =
        \langle f,S g\rangle .
\]

The construction of such  operator $S$ was studied in \cite{mostafaMetric,Ram19}, where a finite-dimensional framework was developed for non-Hermitian dynamics. The central idea is
that the construction of the metric, and hence of the inner product used
to evaluate expectation values, depends on the spectral regime of the
Hamiltonian. 

In the parameter regime corresponding to the spontaneously broken phase, i.e. when the spectrum contains non-degenerate complex-conjugate pairs, the previous positive construction is no longer sufficient. One can instead
construct a self-adjoint intertwining operator by pairing eigenvectors
associated with conjugate eigenvalues. Schematically, such an operator
has the form
\begin{equation}
        S=
        \sum_{j\leq i}
        \delta(\overline E_j-\overline E_i^{\,*})
        \left(
        \alpha_j
        |L_j\rangle
        \langle L_i|
        +
        \alpha_j^*
        |L_i\rangle
        \langle L_j|
        \right),
\end{equation}
with suitable complex coefficients. This operator restores the
intertwining relation
\begin{equation}
        SA=A^\dagger S,
\end{equation}
but it is not positive definite. Hence it defines an indefinite inner
product rather than a Hilbert-space metric.

This leads naturally to a Krein-space interpretation. Since \(S\) is
self-adjoint, its spectral subspaces can be separated into positive and
negative parts. Thus one may write
\[
        \mathcal H=\mathcal H_+\oplus\mathcal H_-,
\]
and decompose
\[
        S=S_+ + S_-,
\]
where \(S_+\) is positive on \(\mathcal H_+\) and \(S_-\) is negative on
\(\mathcal H_-\). A positive operator can then be obtained by reversing
the sign on the negative subspace:
\begin{equation}
        K=S_+-S_- .
\end{equation}
Equivalently, if
\begin{equation}
        S=V D V^\dagger,
\end{equation}
with \(D\) real and nonsingular, one may define
\begin{equation}
        K=V|D|V^\dagger ,
\end{equation}
where \(|D|\) is obtained by replacing each eigenvalue of \(D\) by its
absolute value. The operator \(K\) is self-adjoint and positive
definite, and therefore defines the Hilbert space
\[
        \mathcal H_{K}
        =
        \bigl(\mathcal H,\langle\cdot,\cdot\rangle_{K}\bigr),
        \qquad
        \langle x,y\rangle_{K}
        =
        \langle x,K y\rangle .
\]
However, the positivity of \(K\) should not be confused with
pseudo-Hermiticity of \(A\) with respect to \(K\). In general,
\begin{equation}
        K A\neq A^\dagger K .
\end{equation}
Thus the Krein-positivized metric defines a positive Hilbert-space
geometry, but it does not necessarily make the original Hamiltonian
self-adjoint in that metric. This distinction is crucial when comparing
expectation values, observables and uncertainty relations across metric
representations.

At exceptional points the situation is even more singular. In finite
dimension, an exceptional point is a set in the parameter space at which eigenvalues
and eigenvectors coalesce and the Hamiltonian is no longer diagonalizable.
The operator must then be represented by its Jordan decomposition,
\begin{equation}
        A=\widetilde P B\widetilde P^{-1},
\end{equation}
where the columns of \(\widetilde P\) are generalized eigenvectors  $\{ |\overline{v}_{i} \rangle \} $ and $B$ a Jordan block matrix. In
this case, the metric construction must be formulated in terms of Jordan
chains and dual generalized eigenvectors. One introduces a generalized
symmetry operator \(S_J\), schematically written as
\begin{equation}\label{SB}
        S_B=
        \sum_{j\leq i}
        \delta(\overline E_j-\overline E_i^{\,*})
        \left(
        \alpha_j
        |L_j\rangle
        \langle\overline v_i|
        +
        \alpha_j^*
        |L_i\rangle
        \langle\overline v_j|
        \right),
\end{equation}
which satisfies
\begin{equation}
        S_B A=A^\dagger S_B .
\end{equation}
As in the broken diagonalizable case, \(S_B\) is generally indefinite.
After separating its positive and negative spectral parts, one obtains a
positive operator
\begin{equation}
        S_{K}=S_{B,+}-S_{B,-},
\end{equation}
which defines an inner product adapted to the Jordan representation.

Notice that
\begin{equation}\label{eq:S-KJ}
    S=KJ=JK,
\end{equation}
where $J$ is a fundamental symmetry satisfying $J^2=I$, $J^\dagger=J$,
 and  $[J,K]=0$.
The indefinite form associated with \(S\) can therefore be written as
\begin{equation}\label{eq:indefinite-form-KJ}
    [x,y]_S=  \langle x,Jy\rangle_K.
\end{equation}

\begin{proposition}\label{Jcriterion}
Let \(S=S^\dagger\) be invertible and assume that $SA=A^\dagger S $.
Let \(S=KJ\) be the decomposition defined in \eqref{eq:S-KJ}, with
 \(J=\operatorname{sgn}(S)\). If
\[
        A^{\sharp_K}:=K^{-1}A^\dagger K
\]
denotes the adjoint with respect to the positive operator \(K\), then
\begin{equation}\label{eq:main-identity}
        A^{\sharp_K}=JAJ .
\end{equation}
Consequently, \(K\) is a quasi-Hermitian metric for the original
operator \(A\) if and only if
\begin{equation}\label{eq:J-compatibility}
        [J,A]=0 .
\end{equation}
\end{proposition}

\begin{proof}
The proof is an immediate consequence of the preceding identities. From
\(SA=A^\dagger S\) and the invertibility of \(S\), one has
\[
        A^\dagger=SAS^{-1}.
\]
Using \(S=KJ\), \(J^2=I\), and \(S^{-1}=JK^{-1}\), this gives
\[
        A^\dagger=(KJ)A(JK^{-1}).
\]
Therefore
\[
        A^{\sharp_K}
        =K^{-1}A^\dagger K
        =JAJ .
\]
Hence \(A^{\sharp_K}=A\) holds exactly when \(JAJ=A\), or equivalently when
\([J,A]=0\).
\end{proof}

Finally, if \(A\) is a general non-Hermitian Hamiltonian that is not
pseudo-Hermitian, its spectrum is not necessarily real or arranged in
complex-conjugate pairs, and \(A\) and \(A^\dagger\) are not isospectral
in the pseudo-Hermitian sense. In this case one may still introduce a
positive operator
\begin{equation}
        S_g=
        \sum_{j=1}^{N_{\max}}
        |L_j\rangle
        \langle L_j| ,
\end{equation}
and define
\begin{equation}
        \langle f,g\rangle_{S_g}
        =
        \langle f|S_g g\rangle .
\end{equation}
However,
\begin{equation}
        S_g A\neq A^\dagger S_g .
\end{equation}
Therefore \(S_g\) provides a positive inner product for evaluating
expectation values, but it does not turn \( A\) into a pseudo-Hermitian
operator.

Whatever Hilbert-space structure is chosen, and whatever metric is used to define the corresponding inner product, finite-dimensionality ensures that these metric descriptions can be related. More precisely, all positive definite metrics on a finite-dimensional vector space are connected by congruence. Hence, different metric representations may be identified,
but only after specifying how the associated states and operators are transformed.
We shall make this identification explicit by describing the consistent transport of metrics, states and operators between different representations. This is essential for the analysis of physical quantities, such as expectation values and uncertainty relations

\section{Relation between $\mathcal{H}$ and $\mathcal{H}_{K}$}\label{theorem}

Let \(\eta_1>0\) and \(\eta_2>0\) be the matrix representation of a positive definite sesquilinear form. It is known that there exists \(T\in GL(n)\) such that
\begin{equation}
        \eta_1=T^\dagger\eta_2T .
\end{equation}
In fact, we can take
\begin{equation}
        T=\eta_2^{-1/2}\eta_1^{1/2}.
\end{equation}
or a self-adjoint choice
\begin{equation}\label{eq:T-general}
T=\eta_2^{-1/2}\left(\eta_2^{1/2}\eta_{1}\eta_2^{1/2}\right)^{1/2}\eta_2^{-1/2},
\end{equation}

\begin{proposition}\label{prop1}
Let \(\eta_1>0\) and \(\eta_2>0\) be two metrics on
\(\mathbb C^n\), and $T$ such that
    \begin{equation}
        \eta_1=T^\dagger \eta_2 T
  \end{equation}
Then, the transformation
\[
        x\mapsto \tilde{x}=Tx,
        \qquad
        A\mapsto \tilde{A}:=TAT^{-1}
\]
preserves all metric matrix elements:
\begin{equation}
        \langle \tilde{x},\tilde{A} \tilde{y}\rangle_{\eta_2} =  \langle x,Ay\rangle_{\eta_1}.
    \end{equation}
\end{proposition}

The congruence relation between metrics, together with the corresponding
transformation of states and operators described before,
can be applied to the case of the metric \(K\) with respect to the initial metric structure of the space, defined by the inner product
\(\langle \cdot,\cdot\rangle\).

The metric \(K\) may also be represented in initial Hilbert-space
coordinates. With the spectral notation above, choose
\begin{equation}\label{eq:W-def}
    T:=V \tilde{D}^{-1/2}.
\end{equation}
Then \(T^\dagger K T=I\). Accordingly, a \(K\)-state \(\tilde{x}\) is represented by \(x_{\rm }=T^{-1}\tilde{x}\), and a \(K\)-observable \(A\) by
\(A=T^{-1} \tilde{A}T\). The matrix elements are preserved:
\begin{equation}\label{eq:K-canonical}
    \langle \tilde{x},\tilde{A} \tilde{y}\rangle_K
    =
    \langle T^{-1}x,(T^{-1}AT)T^{-1}y\rangle .
\end{equation}
Thus the passage to initial coordinates is not only a change of coordinates for the metric, but a simultaneous transport of the full representation. By Proposition 1, this transport gives the isomorphic relation between these two Hilbert-spaces. 

Disregarding these linear-algebraic notions, or overlooking the full mathematical framework underlying them, may give rise to conceptual and formal inconsistencies in the analysis of systems governed by non-Hermitian Hamiltonians.

\section{A two-level non-Hermitian spin model}\label{two-level-model}

We now apply the metric-transport constructed in the previous
section to the two-dimensional reduction of the collective-spin
Hamiltonian studied in Ref.~\cite{Ram19}. This model arises as a finite-dimensional non-Hermitian spin Hamiltonian describing a system of \(N\) collective spins interacting through a dissipative one-axis twisting mechanism \cite{kitagawa,Schlosshauer}. In its general form,
the Hamiltonian can be written as
\begin{equation}
    H=-\frac{\omega}{2}+H_{\mathrm{OAT}}+H_d,
\end{equation}
where
\begin{equation}
    H_{\mathrm{OAT}}=-\frac{\lambda}{2N}S_z^2,
    \qquad
    H_d=2i\kappa S_x .
\end{equation}
Here \(S=(S_x,S_y,S_z)\) denotes the collective spin operator of the
system, whose components satisfy the usual angular momentum commutation
relations $ [S_i,S_j]=i\epsilon_{ijk}S_k $.
The corresponding Hilbert space has dimension \(2S+1\), so that the
model provides a natural finite-dimensional setting in which the role of
metric structures can be analyzed explicitly.
The term \(H_{\mathrm{OAT}}\) represents the one-axis twisting interaction,
which is commonly associated with spin squeezing effects. The
term \(H_d\), being non-Hermitian, introduces an effective dissipative
contribution. Depending on the values of the parameters \(\lambda\) and \(\kappa\), the spectrum may be real, may
contain complex-conjugate pairs, or may develop exceptional points. For
this reason, the model is particularly useful for illustrating the
difference between quasi-Hermitian regimes, broken-symmetry regimes and
singular exceptional-point limits.

In what follows, we focus on the two-dimensional representation of this
Hamiltonian. This reduction is sufficient to display explicitly all the
main mechanisms discussed in the previous section. The same formal procedure can be reproduced in higher-dimensional spin representations \cite{alvarez2025uncertainty}.

In a two-dimensional Hilbert space, the model is given by
\begin{equation}\label{eq:H2}
H_2(\omega,\lambda,\kappa)=
\begin{pmatrix}
\dfrac{\lambda}{4}-\dfrac{\omega}{2} & -i\kappa\\[1mm]
-i\kappa & -\dfrac{\lambda}{4}-\dfrac{\omega}{2}
\end{pmatrix}.
\end{equation}
Its eigenvalues are
\begin{equation}\label{eq:eigs}
    E_\pm=-\frac{\omega}{2}\pm\frac{1}{4}\sqrt{\Delta},
    \qquad
    \Delta=\lambda^2-16\kappa^2.
\end{equation}
Thus the model has three different spectral regimes: the exact or
unbroken phase, where \(\Delta>0\); the spontaneously broken phase,
where \(\Delta<0\); and the exceptional point, where \(\Delta=0\).
The aim of this section is not only to construct metrics in these
regimes, but to show how the interpretation of these metrics changes
according to the spectral properties of \(H_2\).

A general Hermitian intertwining operator used in the two-level model can be
written as
\begin{equation}\label{S-param}
S=
\begin{pmatrix}
\beta & i\alpha\\[1mm]
-i\alpha &\quad  \dfrac{\alpha\lambda}{2\kappa}-\beta
\end{pmatrix},
\qquad
\alpha,\beta\in\mathbb R,
\end{equation}
and it satisfies the pseudo-Hermiticity relation
\begin{equation}
    S H_2=H_2^\dagger S .
\end{equation}
The eigenvalues of \(S\) are
\begin{equation}\label{eq:S-eigs}
    \sigma_\pm=\frac{\alpha\lambda\pm\Theta}{4\kappa},
    \qquad
    \Theta=\sqrt{
    16(\alpha^2+\beta^2)\kappa^2
    -8\alpha\beta\kappa\lambda
    +\alpha^2\lambda^2}.
\end{equation}
Therefore the sign of \(S\) depends on the region of parameter space.
This point is essential: when \(S\) is positive, it may define a
Hilbert-space metric compatible with \(H_2\); when \(S\) is indefinite,
its positive modulus defines a geometry, but not
necessarily a quasi-Hermitian representation of the same Hamiltonian.

\subsection{Exact/unbroken phase}

In the exact phase, \(\Delta>0\), the spectrum of \(H_2\) is real. Let
\(\delta=\sqrt{\Delta}\), and the left eigenvectors $|L_{i} \rangle$ are given by
\begin{equation}
|L_{1}\rangle=\{-\frac{i (\delta -\lambda )}{4 \kappa },1\} \quad |L_{2}\rangle= \{ \frac{i (\delta +\lambda )}{4 \kappa },1 \}    
\end{equation}
Then the operator $S$ given in Eq.(\ref{S}), becomes
$$S_{0}= \left(
\begin{array}{cc}
 \frac{\lambda ^2}{4 \kappa ^2}-2 & \frac{i \lambda }{2 \kappa } \\
 -\frac{i \lambda }{2 \kappa } & 2 \\
\end{array}
\right)$$
in coincidence with Eq.(\ref{S-param}), for $\alpha= \frac{\lambda }{2 \kappa }$ and $\beta= \frac{\lambda ^2}{4 \kappa ^2}-2$.

This is the standard quasi-Hermitian situation. In this case, the intertwiner \(S\) in Eq.~\ref{S-param} can be chosen
positive, then \(S\) itself defines another positive
metric for the same real-spectrum Hamiltonian. In that case the two
metric representations, \(I\) and \(S\), are related
by congruence. Namely, there exists \(T\in GL(2)\) such that
\begin{equation}
    T. S.T^{\dagger} =I.
\end{equation}
where 
$$T=\frac{\sqrt{2}}{r \sqrt{\lambda ^4-r^2}}\left(
\begin{array}{cc}
 \kappa  \left(r r_{+}+\Delta ^2 r_{-}\right) & 4 i \kappa ^2 \lambda  r_{-} \\
 -4 i \kappa ^2 \lambda  r_{-} & \kappa  \left(r r_{+}-\Delta ^2 r_{-}\right) \\
\end{array}
\right)$$
being 
\begin{eqnarray}
r&=&\sqrt{256 \kappa ^4-16 \kappa ^2 \lambda ^2+\lambda ^4}\\
r_{\pm}&=& \sqrt{\lambda ^2-r} \pm \sqrt{\lambda ^2+r}
\end{eqnarray}
According to Proposition~\ref{prop1}, the two
descriptions are equivalent only if states and operators are transported
together:
\begin{equation}
    \tilde{x}=Tx,
    \qquad
    \tilde{A}=T AT^{-1}.
\end{equation}
Then, all matrix elements are preserved:
\begin{equation}
    \langle x,Ay\rangle_S
    =
    \left\langle
    T x,
    \left(T A T^{-1}\right)T y
    \right\rangle = 
    \left\langle
    \tilde{x},
    \tilde{A} \tilde{y}
    \right\rangle .
\end{equation}
Thus, in the unbroken phase, the freedom in the choice of a positive
metric does not create an ambiguity if the complete metric
representation is transported consistently. Expectation values,
variances and Robertson-type uncertainty relations are then preserved
under the change of metric representation.

\subsection{Spontaneously broken phase}

We now consider the spontaneously broken phase, \(\Delta<0\). Writing
\begin{equation}
    \delta_b=\sqrt{16\kappa^2-\lambda^2},
\end{equation}
the eigenvalues are
\begin{equation}\label{eq:eigs-broken}
    E_\pm=-\frac{\omega}{2}\pm\frac{i}{2}\delta_b .
\end{equation}
Thus the Hamiltonian has a complex-conjugate pair of eigenvalues. This
already implies that \(H_2\) cannot be self-adjoint with respect to any
positive definite metric. Indeed, if there existed \(\eta>0\) such that
\begin{equation}
    \eta H_2=H_2^\dagger\eta,
\end{equation}
then
\begin{equation}
    h=\eta^{1/2}H_2\eta^{-1/2}
\end{equation}
would be Hermitian in the reference Hilbert space. Consequently, the
spectrum of \(H_2\) would have to be real, contradicting
Eq.~\eqref{eq:eigs-broken}.

This spectral condition is also visible at the level of the
intertwiner \(S\). Since
\begin{equation}
    \Theta^2=(\alpha\lambda-4\beta\kappa)^2+16\alpha^2\kappa^2,
\end{equation}
one has, in the broken region \(|\lambda|<4|\kappa|\) and for
\(\alpha\neq 0\),
\begin{equation}
    \Theta>|\alpha\lambda|.
\end{equation}
Therefore, the eigenvalues of  \(S\) in Eq.(\ref{eq:S-eigs}) have opposite signs
\begin{equation}
    \kappa>0\Rightarrow \sigma_-<0< \sigma_+,
    \qquad
    \kappa<0\Rightarrow \sigma_+ < 0 < \sigma_-.
\end{equation}
and $S$ is not positive definite. It defines instead
an indefinite metric structure.

Following the Krein-type construction discussed in the previous
sections, we diagonalize \(S\) as $S=VDV^\dagger$,
and define its positive modulus by replacing the eigenvalues of \(S\)
with their absolute values:
\begin{equation}\label{eq:SK-model}
    K=V \tilde{D}V^\dagger.
\end{equation}
Explicitly, the diagonalizing unitary matrix may be written as
\begin{equation}
    V=
    \left(
\begin{array}{cc}
 -\dfrac{i (\alpha  \lambda -4 \beta  \kappa +\Theta )}
 {\sqrt{2 \Theta ^2-\Theta  (8 \beta  \kappa -2 \alpha  \lambda )}}
 &
 \dfrac{i (-\alpha  \lambda +4 \beta  \kappa +\Theta )}
 {\sqrt{\Theta  (8 \beta \kappa -2 \alpha  \lambda )+2 \Theta ^2}}
 \\[3mm]
 \dfrac{4 \alpha  \kappa }
 {\sqrt{2 \Theta ^2-\Theta  (8 \beta  \kappa -2 \alpha  \lambda )}}
 &
 \dfrac{4 \alpha  \kappa }
 {\sqrt{\Theta  (8 \beta \kappa -2 \alpha  \lambda )+2 \Theta ^2}}
\end{array}
\right).
\end{equation}
The Krein-positivized metric is then
\begin{equation}\label{eq:SK-explicit}
    K=
    \frac{1}{\Theta}
    \left(
\begin{array}{cc}
 s & i \alpha ^2 \lambda  \\[1mm]
 -i \alpha ^2 \lambda  &
 s+\tfrac{\alpha ^2 \lambda ^2}{2 \kappa }-2 \alpha  \beta  \lambda
\end{array}
\right),
\end{equation}
where $s=4\kappa(\alpha^2+\beta^2)-\alpha\beta\lambda$. 

The main point is that \(K\) is positive, but positivity alone
does not imply quasi-Hermiticity of \(H_2\). Indeed, $K H_2\neq H_2^\dagger K $.
This is exactly the fact identified in Proposition~\ref{Jcriterion}.
Equivalently, if
\begin{equation}
    S=KJ,
    \qquad
    J=\operatorname{sgn}(S)=\left(
\begin{array}{cc}
 -1 & 0 \\
 0 & 1 \\
\end{array}
\right),
\end{equation}
then
\begin{equation}
    H_2^{\sharp K}=JH_2J .
\end{equation}
Thus \(K\) is a quasi-Hermitian metric for \(H_2\) if and only if
\begin{equation}
    [J,H_2]=0.
\end{equation}

The two positive geometries \(I\) and \(K\) can 
be related by the transport rule of
Proposition~\ref{prop1}. If \(T_{\rm br}\in GL(2)\) is
chosen so that
\begin{equation}
    T_{br}K. T_{br}^{\dagger}=I,
\end{equation}
then the corresponding dictionary is
\begin{equation}
    x_{\rm br}=T_{br}x,
    \qquad
    A_{\rm br}=T_{br}AT_{br}^{-1}.
\end{equation}
where 
\begin{equation}
T_{br}=\frac{\sqrt{\kappa }}{\Theta  \sqrt{\Theta ^2-\alpha ^2 \lambda ^2}}  \left(
\begin{array}{cc}
 \Theta_{-} (4 \beta  \kappa -\alpha  \lambda
   )+\Theta  \Theta_{+} & 4 i \alpha  \kappa \Theta_{-} \\
 -4 i \alpha  \kappa  \Theta_{-} & \Theta_{-} (\alpha  \lambda -4 \beta  \kappa )+\Theta \Theta_{+} \\
\end{array}
\right)
\end{equation}
being $\Theta_{\pm}= \sqrt{\Theta -\alpha  \lambda } \pm \sqrt{\Theta +\alpha  \lambda}$
Thus
\begin{equation}
    \langle x,Ay\rangle_{K}
    =
    \left\langle
    T_{br}x,
    \left(T_{br}AT_{br}^{-1}\right)T_{br}y
    \right\rangle = 
    \left\langle
    x_{br},
    A_{br}y_{br} \right\rangle.
\end{equation}
This equation is the concrete form, for the broken two-level model, of
the general metric-transport principle. It shows that different metrics
can be compared only if the states and the operators are transported
together. If one changes the metric but keeps the same coordinate
representatives of states and observables, one is mixing different
Hilbert-space geometries. Apparent anomalies in expectation values,
variances or Robertson-type uncertainty relations may then be artifacts
of this mismatch rather than genuine physical effects \cite{mandel}.

\subsection{Exceptional point}

At the exceptional point the Hamiltonian is not diagonalizable and
its eigenvectors coalesce \cite{Starkov}. We consider the boundary
\begin{equation}
    \Delta=\lambda^2-16\kappa^2=0,
\end{equation}
and, for simplicity, take $\lambda=4\kappa$ and $\kappa\neq 0$.
The Hamiltonian and its adjoint are then
\begin{equation}
H_{\rm EP}
=
\begin{pmatrix}
\kappa-\dfrac{\omega}{2} & -i\kappa\\[1mm]
-i\kappa & -\kappa-\dfrac{\omega}{2}
\end{pmatrix},
\qquad
H_{\rm EP}^{\dagger}
=
\begin{pmatrix}
\kappa-\dfrac{\omega}{2} & i\kappa\\[1mm]
i\kappa & -\kappa-\dfrac{\omega}{2}
\end{pmatrix}.
\end{equation}
Both eigenvalues collapse into
\begin{equation}
    E_{\rm EP}=-\frac{\omega}{2},
\end{equation}
and \(H_{\rm EP}\) has a second-order exceptional point.

The Hermitian intertwining operator obtained from Eq.~\eqref{SB} at the
exceptional point is
\begin{equation}
S_{\rm EP}
=
\begin{pmatrix}
\beta & i\alpha\\[1mm]
-i\alpha & -\beta+2\alpha
\end{pmatrix}.
\end{equation}
Its eigenvalues are
\begin{equation}
        s_\pm
        =
        \alpha\pm\sqrt{\alpha^2+(\beta-\alpha)^2}.
\end{equation}
Hence, for \(\alpha\neq\beta\), \(S_{\rm EP}\) is indefinite. As in the
broken phase, one may pass from this indefinite form to a positive
operator by Krein positivization.

For \(\alpha\neq\beta\), a positive geometry can be obtained from the
Krein-positivized form. It can be related to the standard Hilbert-space
geometry by a congruence transformation. In the present two-dimensional
case one may choose
\begin{equation}
T= \frac{1}{2q(\alpha - \beta)}
\begin{pmatrix}
 q q_{+}-(\alpha -\beta ) q_{-} & i \alpha  q_{-} \\[1mm]
 -i \alpha  q_{-} & \alpha  q_{+}-\beta  (q_{+}+q_{-})
\end{pmatrix},
\end{equation}
where
\begin{equation}
q= \sqrt{2 \alpha ^2-2 \alpha  \beta +\beta ^2},
\qquad
q_{\pm}= \sqrt{q - \alpha} \pm  \sqrt{q + \alpha}.
\end{equation}
This operator implements the metric transport between the positive
geometry associated with the exceptional-point limit and the reference
Hilbert space, provided the same transport is applied simultaneously to
states and operators.

\section{Uncertainty relation under metric transport}\label{uncertainty}

We now examine uncertainty relations from the viewpoint of metric
transport \cite{uncert,alvarez2025uncertainty,MacconePati2014,URNH}.
The aim is not to derive a new uncertainty principle for non-Hermitian
systems, but to emphasize that uncertainty inequalities are
representation-consistent only when expectation values, variances, states,
adjoints and observables are evaluated within the same metric structure.

The standard Heisenberg--Robertson relation bounds the product of
variances from below by the expectation value of the commutator. However,
as was emphasized in Ref.~\cite{MacconePati2014}, product-form uncertainty
relations may become trivial when one of the variances vanishes, even if
the observables are not jointly sharp on the state. This motivates stronger
relations involving sums of variances, such as
\(\Delta A^2+\Delta B^2\), whose lower bounds may remain nontrivial in
these limiting situations.

Uncertainty relations for non-Hermitian operators have been discussed from
several viewpoints, including formulations based on metric-compatible or
``good'' observables \cite{mandel,Bag23,URNH}. The point emphasized here is
complementary. Even standard Robertson-type inequalities, or their stronger
sum-variance versions, may appear to fail if the metric is changed while
the same coordinate representatives of states and observables are kept
fixed. Such an apparent failure does not indicate a violation of the
uncertainty principle, but rather a mismatch between different metric
representations.

Let \(\eta_1>0\), \(\eta_2>0\), and let \(T\in GL(n)\) be such that
\begin{equation}
    \eta_1=T^\dagger\eta_2T .
\end{equation}
Under the simultaneous transport
\begin{equation}\label{diccionario}
    \tilde{x}=Tx,
    \qquad
    \tilde{A}=TAT^{-1},
    \qquad
    \tilde{B}=TBT^{-1},
\end{equation}
the identities derived in the previous section give
\begin{equation}
    \langle A\rangle_{\eta_1,x}
    =
    \langle \tilde{A}\rangle_{\eta_2,\tilde{x}},
    \qquad
    \Delta^{2}_{\eta_1,x}A
    =
    \Delta^{2}_{\eta_2,\tilde{x}}\tilde{A},
\end{equation}
and similarly for \(B\). Here, expectation values are always understood as
normalized metric expectation values:
\begin{equation}
    \langle A\rangle_{\eta,x}
    =
    \frac{\langle x,\eta A x\rangle}{\langle x,\eta x\rangle}.
\end{equation}
The commutator and covariance terms are transported in the same way:
\begin{equation}
    \langle[A,B]\rangle_{\eta_1,x}
    =
    \langle[\tilde A,\tilde B]\rangle_{\eta_2,\tilde x}.
\end{equation}
Therefore,
\begin{equation}
    \Delta_{\eta_1,x} A\,\Delta_{\eta_1,x} B
    \geq
    \frac12
    \left|
    \langle[A,B]\rangle_{\eta_1,x}
    \right|
\end{equation}
holds if and only if
\begin{equation}
    \Delta_{\eta_2,\tilde{x}} \tilde{A}\,
    \Delta_{\eta_2,\tilde{x}} \tilde{B}
    \geq
    \frac12
    \left|
    \langle[\tilde{A},\tilde{B}]\rangle_{\eta_2,\tilde{x}}
    \right|
\end{equation}
is valid in the transported representation.

The same observation applies to sum uncertainty relations. The Maccone-Pati type inequality in the Hilbert space $\mathcal{H}_{\eta}$ is given by 
\cite{MacconePati2014}
\begin{equation}
    \Delta_{\eta,x}^2 A+\Delta_{\eta,x}^2 B
    \geq
    \pm i\langle[A,B]\rangle_{\eta,x}
    +
    \left|
    \langle x,(A\pm iB)x^\perp\rangle_{\eta}
    \right|^2.
\end{equation}
Here
\(x^\perp\) satisfies  $ \langle x,x^\perp\rangle_\eta=0$,
and the sign is chosen so that the commutator contribution is
non-negative. A second sum-type bound is
\begin{equation}
    \Delta_{\eta,x}^2 A+\Delta_{\eta,x}^2 B
    \geq
    \frac12
    \left|
    \langle x^\perp_{A+B},(A+B)x\rangle_{\eta}
    \right|^2 ,
\end{equation}
where
\begin{equation}
    x^\perp_{A+B}
    =
    \frac{
    (A+B-\langle A+B\rangle_{\eta,x})x
    }{
    \Delta_{\eta,x}(A+B)
    }
\end{equation}
whenever \(\Delta_{\eta,x}(A+B)\neq0\). Equivalently, this last inequality
can be written as
\begin{equation}
    \Delta_{\eta,x}^2 A+\Delta_{\eta,x}^2 B
    \geq
    \frac12
    \Delta_{\eta,x}^2(A+B).
\end{equation}
As in the case of the Robertson inequality, the construction of the new Hilbert space endowed with $\eta_{2}$ through Eq.~\ref{diccionario}, together with the corresponding transport of states and operators, guarantees that analogous sum-uncertainty relations remain valid in the transported metric representation. 

\subsection{Example}

The purpose of this example is not to construct a counterexample to the
Robertson--Heisenberg uncertainty relation. Rather, it is to show explicitly
how an apparent violation may arise when the metric is changed but the
coordinate representatives of the state and of the operators are kept fixed.
In such a situation, the quantities entering the uncertainty relation are no
longer evaluated within a single metric representation.

We now specialize the example discussed in Section~\ref{two-level-model} to the broken-symmetry
region by considering the particular parameter choice \(\lambda=12\kappa/5\).
For this specific case, the matrix constructed in Eq.~\eqref{eq:SK-explicit}
is given by

$$
K(\alpha,\beta)=
\frac{1}{\sqrt{34\alpha^2-30\alpha\beta+25\beta^2}}
\begin{pmatrix}
5\alpha^2-3\alpha\beta+5\beta^2 & 3i\alpha^2 \\[1mm]
-3i\alpha^2 &
\frac{43\alpha^2-45\alpha\beta+25\beta^2}{5}
\end{pmatrix}.
$$

For compactness, we introduce

$$
u=\sqrt{25\alpha^2-30\alpha\beta+25\beta^2},
\qquad
v=\sqrt{34\alpha^2-30\alpha\beta+25\beta^2},
$$
$$
c=5\alpha^2-3\alpha\beta+5\beta^2,
\qquad
d=43\alpha^2-45\alpha\beta+25\beta^2.
$$

We take

$$
A=\sigma_x,
\qquad
B=\sigma_y,
\qquad
\psi_K=N_K(x,y)^T,
\qquad
N_K=\left(\frac{c_\psi}{5v}\right)^{-1/2}
=\sqrt{\frac{5v}{c_\psi}},
\qquad
x,y\in\mathbb{R},
$$

with $x^2+y^2\neq 0$. With these definitions, the $K$-norm of the fixed
coordinate vector is

$$
\langle \psi_K,\psi_K\rangle_K=1,
\qquad
c_\psi=5cx^2+dy^2.
$$

Let us first perform the computation in the inconsistent way. Namely, we
change the metric from the canonical one to $K(\alpha,\beta)$, but keep
fixed the same coordinate representatives $A$, $B$, and $\psi$. Using the
$K$-metric expectation values, one finds

$$
\langle A\rangle_{K,\psi_K}=
\frac{
15i\alpha^2(x^2-y^2)+2v^2xy
}
{c_\psi}.
$$

Similarly,

$$
\langle B\rangle_{K,\psi_K}=
\frac{
-15\alpha^2(x^2+y^2)
+
6i\alpha(3\alpha-5\beta)xy
}
{c_\psi}.
$$

These expectation values are not, in general, real. This already indicates
that the fixed Pauli representatives $A=\sigma_x$ and $B=\sigma_y$ are not
observables in the Hilbert space endowed with the metric $K$.

The corresponding variances are

$$
\Delta_{K,\psi_K}^2 A=
\frac{
u^2v^2(x^2-y^2)^2
}
{c_\psi^2}.
$$

For the second observable one obtains

$$
\Delta_{K,\psi_K}^2 B=
\frac{
u^2v^2(x^2+y^2)^2
}
{c_\psi^2}.
$$

Therefore,

$$
\Delta_{K,\psi_K}A\Delta_{K,\psi_K}B=
\frac{
u^2v^2|x^4-y^4|
}
{c_\psi^2}.
$$

On the other hand, the commutator term evaluated with the same fixed
representatives gives

$$
\frac12
\left|
\langle[A,B]\rangle_{K,\psi_K}
\right|=
\frac{
\sqrt{
900\alpha^4x^2y^2+
\left(5cx^2-dy^2\right)^2
}
}
{c_\psi}.
$$

Equivalently, the difference between the squared left-hand side and the
squared right-hand side of the Robertson product inequality is

$$
\left(
\Delta_{K,\psi_K}A\Delta_{K,\psi_K}B
\right)^2
-\frac14
\left|
\langle[A,B]\rangle_{K,\psi_K}
\right|^2
=
\frac{
u^4v^4(x^4-y^4)^2
}
{c_\psi^4}
-\frac{
900\alpha^4x^2y^2+
\left(5cx^2-dy^2\right)^2
}
{c_\psi^2}.
$$

This quantity can be negative for suitable values of the parameters and of
the state. Thus, in the non-transported computation one may obtain the
apparent violation

$$
\Delta_{K,\psi_K}A\Delta_{K,\psi_K}B
<
\frac12
\left|
\langle[A,B]\rangle_{K,\psi_K}
\right|.
$$

This apparent violation occurs because the operators $A$ and $B$ have not
been transported together with the metric. In particular, with respect to
the $K$-metric adjoint

$$
X^{\sharp_K}=K^{-1}X^\dagger K,
$$

the fixed Pauli representatives are not, in general, $K$-self-adjoint:

$$
A^{\sharp_K}\neq A,
\qquad
B^{\sharp_K}\neq B.
$$

Thus the Robertson-Heisenberg inequality is being applied outside the
metric representation in which its hypotheses are satisfied.

Let us now repeat the computation consistently. We use the metric transport
operator $T$ defined by

$$
\eta_1=T^\dagger K T,
$$

where $\eta_1$ denotes the initial metric representation. In the present
case, the initial representation is the canonical one, so one may take
$\eta_1=I$. The state and operators must be transported simultaneously
according to

$$
\psi_0=\frac{1}{\sqrt{x^2+y^2}}(x,y)^T,
\qquad
\widetilde{\psi}=T\psi_0,
\qquad
\widetilde{A}=TAT^{-1},
\qquad
\widetilde{B}=TBT^{-1}.
$$

In the $K$-metric representation, the corresponding adjoints are

$$
\widetilde{A}^{\sharp_K}=K^{-1}\widetilde{A}^\dagger K,
\qquad
\widetilde{B}^{\sharp_K}
=K^{-1}\widetilde{B}^\dagger K.
$$

A direct calculation gives

$$
\widetilde{A}^{\sharp_K}=\widetilde{A}
=\frac{1}{uv}
\begin{pmatrix}
-15i\alpha^2 &
d\\
5c &
15i\alpha^2
\end{pmatrix}.
$$

Likewise,

$$
\widetilde{B}^{\sharp_K}=\widetilde{B}
=\frac{1}{v^2(u^2+uv)}
\begin{pmatrix}
45\alpha^3(5\beta-3\alpha) &
\frac{i}{2}
\left[
225\alpha^4-
\left(d+uv\right)^2
\right]\\
-\frac{i}{2}
\left[
225\alpha^4-
\left(5c+uv\right)^2
\right]
&
45\alpha^3(3\alpha-5\beta)
\end{pmatrix}.
$$

Therefore, the transported operators are genuine observables in the
transported metric representation.

The corresponding expectation values are

$$
\langle \widetilde{A}\rangle_{K,\widetilde{\psi}}=
\frac{2xy}{x^2+y^2},
\qquad
\langle \widetilde{B}\rangle_{K,\widetilde{\psi}}=0.
$$

Moreover, the variances are

$$
\Delta_{K,\widetilde{\psi}}^2\widetilde{A}=
\frac{(x^2-y^2)^2}{(x^2+y^2)^2},
\qquad
\Delta_{K,\widetilde{\psi}}^2\widetilde{B}=1.
$$

Thus

$$
\Delta_{K,\widetilde{\psi}}\widetilde{A}
\Delta_{K,\widetilde{\psi}}\widetilde{B}
=
\frac{|x^2-y^2|}{x^2+y^2}.
$$

The commutator term is

$$
\frac12
\left|
\langle[\widetilde{A},\widetilde{B}]
\rangle_{K,\widetilde{\psi}}
\right|
=
\frac{|x^2-y^2|}{x^2+y^2}.
$$

Consequently,

$$
\Delta_{K,\widetilde{\psi}}\widetilde{A}
\Delta_{K,\widetilde{\psi}}\widetilde{B}
=
\frac12
\left|
\langle[\widetilde{A},\widetilde{B}]
\rangle_{K,\widetilde{\psi}}
\right|.
$$

With the transported state and operators, the Robertson--Heisenberg
relation is restored. In this example, the transported computation saturates
the Robertson product inequality.

This completes the comparison between the inconsistent and the transported
computations. The apparent failure of the Robertson product form comes from
mixing inequivalent metric representations.


\section{Conclusions}\label{conclusiones} 

We have analyzed the consistency of metric representations in finite-dimensional non-Hermitian quantum mechanics. The main point of the work is to stress that the physical description of the system can be given in different Hilbert spaces, through the construction of an isomorphism. 

For pseudo-Hermitian operators, we considered  theirKrein decomposition
\[
S=KJ,
\]
where \(K=|S|\) is positive and \(J=\operatorname{sgn}(S)\) is the associated fundamental symmetry. Although \(K\) defines a positive Hilbert-space geometry, it does not
necessarily define a quasi-Hermitian metric for the original Hamiltonian. Indeed, the \(K\)-adjoint of \(H\) is given by
\[
H^{\sharp_K}=JHJ,
\]
and therefore \(K\) is compatible with \(H\) as a quasi-Hermitian metric if and only if
\[
[J,H]=0.
\]
Nevertheless, $K$ is a metric and can be used to define a new Hilbert space $\mathcal{H}_{K}$.

In Section~\ref{theorem} explicit isomorphism between the Hilbert space endowed  with the Krein-positivized metric $K=|S|$ and the usual Hilbert space has been constructed.
That means that both Hilbert spaces can be used in defining the dynamics of non-Hermitian systems.

In Section~\ref{two-level-model} a two-level non-Hermitian spin model has been used to illustrate the main statements of our theoretical discussion. In the unbroken regime, where the spectrum is real, a metric may be used to obtain a quasi-Hermitian representation of the Hamiltonian. In the spontaneously broken regime and at the exceptional points, the Krein-positivized operator \(K\) still defines a positive geometry, but the Hamiltonian is not \(K\)-self-adjoint in the parameter space of the model.

We also used the Robertson uncertainty relation as consistency test. Once states, observables, and their adjoints are evaluated within a single metric Hilbert space, the uncertainty relations retain their physical meaning.

In this sense, our results provide a geometric interpretation of the good-observable prescription of Ref.~\cite{mandel}: good observables are precisely the operators that belong to a given metric representation, or equivalently those obtained by transporting ordinary Hermitian observables through the Hilbert-space isomorphism. This shows that the validity of uncertainty relations is not tied to a fixed coordinate representation, but to the consistent transport of the metric, the states, the adjoint operation and the observables.

The finite-dimensional setting considered here avoids domain subtleties and makes the transport mechanism explicit. In infinite-dimensional systems, additional conditions on domains, boundedness, invertibility, and metric dynamics are required. 

\appendix

\section{}

The purpose of the following lemma is to separate two issues that are often 
identified too quickly. Although the Krein-positivized operator \(K=|S|\) is a 
positive metric, the original non-Hermitian Hamiltonian need not be 
\(K\)-self-adjoint. Hence the admissible observables in the Hilbert space 
\(\mathcal H_K\) must be characterized independently. The lemma below provides 
this characterization for the two-dimensional broken-symmetry metric used in 
the main text.

\begin{lemma}\label{lema}
Let \(K\) be the positive operator obtained from the indefinite
intertwiner \(S\) in the broken region \(\Delta<0\). A matrix
\[
\widehat H=
\begin{pmatrix}
a_{11} & a_{12}\\
a_{21} & a_{22}
\end{pmatrix}
\]
is \(K\)-self-adjoint, namely
\begin{equation}
    K\widehat H=\widehat H^\dagger K,
\end{equation}
if and only if it belongs to the following real four-parameter family:
\begin{align}
a_{11}
&=
a-\frac{2 i \alpha ^2 b \kappa  \lambda }
{8 \kappa ^2 \left(\alpha ^2+\beta ^2\right)+\alpha ^2 \lambda ^2-6 \alpha  \beta  \kappa  \lambda },
\\[0.4em]
a_{12}
&=
b+i b_{1},
\\[0.4em]
a_{21}
&=
\frac{2 b \kappa  \left(4 \kappa  \left(\alpha ^2+\beta ^2\right)-\alpha  \beta  \lambda \right)}
{8 \kappa ^2 \left(\alpha ^2+\beta ^2\right)+\alpha ^2 \lambda ^2-6 \alpha  \beta  \kappa  \lambda }
+i c_{1},
\\[0.4em]
a_{22}
&=
2a-a_{11}
-\frac{4 \kappa  \left(\alpha ^2+\beta ^2\right) (b_{1}+c_{1})}{\alpha ^2 \lambda }
+\frac{\beta  (b_{1}+3 c_{1})}{\alpha }
-\frac{c_{1} \lambda }{2 \kappa },
\end{align}
where \(a,b,b_1,c_1\in\mathbb R\).
\end{lemma}

\begin{proof}
The result follows by inserting a generic matrix
\(\widehat H=(a_{ij})_{i,j=1}^2\) into
\(K\widehat H=\widehat H^\dagger K\) and solving the resulting
linear system for its entries. The displayed parametrization is valid
under the non-degeneracy assumptions
\(\alpha\neq 0\), \(\lambda\neq 0\), \(\kappa\neq 0\), and
\(\Theta\neq 0\). The excluded cases correspond to degenerate limits in
which the metric geometry becomes trivial, the Hamiltonian is Hermitian,
or the original intertwiner is already positive definite.
\end{proof}

This lemma identifies the class of operators that are self-adjoint in
the positive geometry defined by \(K\). In particular, inserting the
entries of the original Hamiltonian \(H_2\) into
Lemma~\ref{lema} gives the condition
\begin{equation}\label{eq:condition-H2-SK}
    \kappa=
    \frac{\alpha\beta\lambda}
    {2(\alpha^2+\beta^2)}.
\end{equation}
Under this condition,
\begin{equation}
\Delta
=
\lambda^2-16\kappa^2
=
\lambda^2
\left[
1-\frac{4\alpha^2\beta^2}{(\alpha^2+\beta^2)^2}
\right]
=
\lambda^2
\frac{(\alpha^2-\beta^2)^2}{(\alpha^2+\beta^2)^2}
\geq 0.
\end{equation}
Therefore \(H_2\) can be \(K\)-self-adjoint only outside the 
broken region. Except for the borderline case
\(\alpha^2=\beta^2\), where \(\Delta=0\), condition
\eqref{eq:condition-H2-SK} places the model in the unbroken regime
\(\Delta>0\). In that regime, however, the original intertwining
operator \(S\) may already be chosen positive, and the Krein
positivization is not needed.


%

\end{document}